\documentclass[11pt,letter]{article}
\usepackage{}
\usepackage{bbding}
\fussy
\usepackage{framed, xcolor}
\usepackage{tocvsec2}
\usepackage{datetime}
\usepackage{pifont}
\usepackage{pdflscape}
\usepackage{subfigure}          
\usepackage{colortbl}
\usepackage{booktabs}
\usepackage{pdfsync}
\usepackage[font=scriptsize,bf]{caption}
\usepackage{tikz,subfigure}
\usepackage[active]{srcltx}
\usepackage[margin=1in]{geometry}
\usepackage{algorithm}
\usepackage{algorithmic}
\usepackage{epsfig,amssymb,amsfonts,amsmath,amsthm}
\usepackage{multirow}
\usepackage[numbers,sort&compress,sectionbib]{natbib}
\usepackage{amsfonts}
\usepackage{xspace}
\usepackage{tabularx}
\usepackage{pstricks}
\usepackage{setspace}
\usepackage{xcolor}

\usepackage{rotating}
\usepackage{minitoc}

\usepackage{hyperref}

\usepackage{tikz}
\usepackage{enumerate}
\usetikzlibrary{chains,fit,shapes,arrows}
\usetikzlibrary{shapes.arrows}

\newcommand{\remove}[1]{}

\newcommand{\ce}{\mathrm{e}}

\newcommand{\size}{\mathrm{size}}

\renewcommand{\deg}{\mathrm{deg}}

\renewcommand{\leq}{\leqslant}
\renewcommand{\geq}{\geqslant}

\newcommand{\thmref}[1]{Theorem~\ref{thm:#1}}

\newcommand{\lemref}[1]{Lemma~\ref{lem:#1}}

\newcommand{\eq}[1]{\eqref{eq:#1}}

\newcommand{\Pro}[1]{\mathbf{Pr} \left[\,#1\,\right]}

\newcommand{\e}{\mathbf{E}}

\renewcommand{\epsilon}{\varepsilon}

\newtheorem{thm}{Theorem}  

\newtheorem{lem}[thm]{Lemma}

\newtheorem{cor}[thm]{Corollary}

\newtheorem{pro}[thm]{Proposition}

\def\ce{\mathrm{e}}

\def\auto{\mathrm{auto}}

\def\auto{\mathrm{auto}}

\usepackage{setspace}

\newcommand{\ee}{\mathbf{E}}

\title{{ Counting Hypergraphs in Data Streams}}

\author{He Sun\\
Max Planck Institute for Informatics\\
Saarbr\"{u}cken, Germany\\
\texttt{hsun@mpi-inf.mpg.de}}
\date{}

\begin{document}

\maketitle
\begin{abstract}
We  present the first streaming algorithm for counting  an arbitrary hypergraph $H$ of constant size in a massive hypergraph $G$. Our algorithm can handle both edge-insertions and edge-deletions, and is applicable for the distributed setting. Moreover, our approach provides the first family of graph polynomials for the hypergraph counting problem. Because of the close relationship between hypergraphs and set systems, our approach may have applications in studying similar problems.

\end{abstract}

\section{Introduction}

The problem of counting subgraphs is one of the fundamental questions in algorithm design, and has various applications in analyzing the clustering and transitivity coefficients of networks, uncovering structural information of graphs that model biological systems, and designing graph databases. While the exact counting of subgraphs of constant size is polynomial-time solvable, traditional algorithms need to store the whole graph and compute the solution in an off-line fashion, which is not practical even for graphs of medium size. A modern way to deal with this problem is to design algorithms in the  streaming setting, where the edges of the underlying graph come sequentially in an arbitrary order, and algorithms with sub-linear space are required to approximately count the number of occurrences of certain subgraphs. Since the first streaming $\mbox{algorithm}$ by \citet{conf/soda/Bar-YossefKS02}, this problem has received much attention in recent years~\cite{conf/cocoon/JowhariG05,conf/esa/BuriolFLS07,
conf/esa/ManjunathMPS11,conf/icalp/KaneMSS12,
conf/icdm/BordinoDGL08,conf/kdd/BecchettiBCG08,conf/soda/Bar-YossefKS02,
conf/pods/AhnGM12}.

We address the subgraph counting problem for  hypergraphs. Formally, we are given a sequence of sets
$s_1,s_2,\ldots$ in a data stream. These sets,
each of which consisting of vertices of the underlying hypergraph $G$, arrive sequentially and represent  edges of a hypergraph  $G=(V,E)$.
Moreover, every coming edge $e_i$ is equipped with a  sign~(``$+$" or ``$-$"), indicating that edge $e_i$ is inserted to or deleted from the hypergraph $G$. That is, we study the so-called \emph{turnstile model}~\cite{DBLP:journals/fttcs/Muthukrishnan05} where the underlying graph may change over time. For any hypergraph $H$ of constant size, algorithms with sub-linear space are required to approximate the number of occurrences of $H$ in $G$.

%
%
%

\paragraph{Motivation.}  Hypergraphs are  basic models to characterize precise relations $\mbox{among}$ items of data sets. For the study of databases, people started to use $\mbox{hypergraphs}$ to model database schemes since 1980s~\cite{journals/jacm/Fagin83,conf/stoc/BeeriFMMUY81}, and this line of research led to several well-known  data storage mechanisms like \textsf{HyperGraphDB}~\cite{hypergraphDB}.
Besides  database theory, a number of  studies have shown that simple graphs\footnote{For ease of our discussion simple graphs refer to graphs where every edge consists of two vertices.},
representing pairwise relationships,
 are usually not  sufficient to encode all information when studying social, protein, or biological networks, and suggested to use hypergraphs to model the real relations among the items.
 For illustrating this point of view, let us look at the coauthor network for example. In a coauthor network, authors are represented as vertices of a graph, and an edge between two authors exists iff these two persons are co-authors. This natural model misses the information on whether a set of three (or more) authors have been co-authored of the same article. Such information loss is undesirable for many applications, e.g., for detecting communities or clusters like all authors that worked in the same research area.
 Similar problems occur in studying biological, social, and other networks when hypergraphs are required in order to express the complete relation among entities~\cite{journals/PCB/KUT09,conf/nips/ZhouHS06}.

%
%
%

\paragraph{Our Results \& Techniques.}
 We initiate the study of counting subgraphs in the streaming setting, and present the first algorithm for this problem.
 Although the subgraph counting problem is much more difficult for the case of hypergraphs and streaming algorithms were unknown even for the edge-insertion case prior to our work,  our algorithm runs in the general turnstile model, and is applicable in the distributed setting.
Formally, for any fixed subgraph $H$ of constant size, our algorithm $(1\pm\varepsilon)$-approximates the number of occurrences of $H$ in $G$. That is, for any constant $\varepsilon\in(0,1)$, the output of our algorithm satisfies $Z\in[(1-\varepsilon)\cdot\#H, (1+\varepsilon)\cdot\#H]$ with probability at least $2/3$.   The main result of our paper is as follows:

\begin{thm}[Main Result]\label{thm:main}
Let $G$ be a hypergraph of $n$ vertices and $m$ edges, and $H$ a hypergraph of $k$ edges and minimum degree at least 2. Then there is an algorithm to $(1\pm\varepsilon)$-approximate the number of occurrences of $H$ in $G$ that uses $O\left(\frac{1}{\varepsilon^2}\cdot\frac{m^k}{(\#H)^2}\cdot\log n\right)$ bits of space. The update time per coming edge is $O\left(\frac{1}{\varepsilon^2}\cdot\frac{m^k}{(\#H)^2}\right)$.
Our algorithm works in the turnstile model.
\end{thm}

To compare our algorithm with  na\"ive methods, note that a na\"ive approach for counting  $\# H$
needs to
either sample independently $k$ vertices (if possible) or $k$ edges
from the stream. Since the probability of $k$ vertices (or $k$ edges) forming $H$ is $\# H/{n^k}$  (or $\# H/{m^k}$), this approach needs space
$\Omega\left(\frac{n^k \log n}{ \#H} \right)$ and $\Omega\left(\frac{m^k \log
n}{\#H}\right)$, respectively. Thus our algorithm has significant improvement over the na\"ive approach.
On the other hand,  we note that for any graph $G$ of $m$ edges, and hypergraph $H$ of $k$ edges, the number of $H$ in $G$ can be as big as $\Omega(m^{k/2})$. Hence for dense graphs with $\# H=\omega\left(m^{\frac{k-1}{2}}\right)$, our algorithm achieves a $(1+\varepsilon)$-approximation in sublinear space.

Our algorithm uses the composition of complex-valued random variables.
Besides presenting the first  hypergraph counting algorithm in the streaming setting,  our approach yields a family of graph polynomials $\{p_H\}$ to count the number of hypergraph $H$ in  hypergraph $G$. That is, for any hypergraph $H$ the polynomial $p_H$ takes hypergraph $G$ as an argument, and the value of $p_H(G)$ is the number of isomorphic copies of $H$ in $G$.
This is the first family of graph polynomials for the hypergraph counting problem, and  the techniques developed here  may have applications
 in studying graph theory or related topics.

\begin{thm}\label{thm:graphpoly}
For any hypergraph $H$, there is a graph polynomial $p_H(\cdot)$ such that for any hypergraph $G$, $p_H(G)\in\mathbb{N}\cup\{0\}$ is the number of isomorphic copies of $H$ in $G$.
\end{thm}

Our algorithm follows the framework by \citet{conf/icalp/KaneMSS12}. For any hypergraph $H$ of $k$ edges, we  maintain $k$ variables $Z_{e_1^{\star}},\ldots, Z_{e^{\star}_k}$, and each variable $Z_{e^{\star}_i}$ corresponds to  one edge  in $H$. For every coming edge $e$ in graph $G$, we choose one or more $Z_{e_i^{\star}}$ to update according to the value of hash functions. We will prove that the returned value of $\prod_{1\leq i\leq k} Z_{e_i^{\star}}$ is unbiased.
However, in contrast to the simple graph case, the algorithm for hypergraphs and the analysis is much more complicated due to the following reasons:
\begin{enumerate}
\item In contrast to simple graphs, subgraph isomorphoism between hypergraphs is more difficult to handle, and hence
the update procedure for every coming edge is more involved. To overcome this, for every coming edge $e$ of hypergraph $G$ that consists of $\ell$ edges, we look at $\ell !$ permutations of $\{1,\ldots,\ell\}$, and every such permutation gives $e$ an ``orientation".
Moreover, instead of  updating every $Z_{e^{\star}_i}$ simultaneously for the simple graph case, we
choose one or more $Z_{e^{\star}_i}$ to update. Through this, we  prove that the returned value of our estimator is unbiased for the number of occurrences of $H$ in $G$.

\item The second difficulty for dealing with hypergraphs comes from analyzing the concentration of the estimator.  All previous works on the subgraph counting problem, e.g.~\cite{conf/cocoon/JowhariG05,conf/esa/ManjunathMPS11,conf/icalp/KaneMSS12}, indicate that the space requirement of the algorithm depends on the number of other subgraphs in the underlying graph. For instance, the space complexity of the algorithms by \cite{conf/esa/ManjunathMPS11,conf/icalp/KaneMSS12,conf/cocoon/JowhariG05} is essentially determined by the number of closed walks of certain length in  graph $G$. However, the notion of closed walks in (non-uniform) hypergraphs is not well-defined, and hence we need to use alternative methods to analyze the concentration of the estimator, as well as the space requirement.
\end{enumerate}
 Because of these differences, our generalization is non-trivial and elegant.  Our result (\thmref{main}) shows that the regularity of hyperedges in  $G$ and $H$ does not influence the actual space complexity of the algorithm, and the time and space complexity of our algorithm is the same as the simple graph case.

\paragraph{Notation.}
Let $G=(V,E)$ be a hypergraph graph. The set of vertices and edges are represented by $V[G]$ and $E[G]$. We  assume that graph $G$ has $n$ vertices, and $n$ is known in advance.
Graph $G$ is called a hypergraph if every edge $e\in E[G]$ is a non-empty subset of $V[G]$, i.e. $E[G]$ is a subset of the power set of $V[G]$.
For any hypergraph $G$ and  vertex $u\in V[G]$, the degree of $u$, expressed by $\deg(u)$, is the number of edges that include $u$.
Moreover, the size of edge $e\in E[G]$, denoted by $\mathrm{size}(e)$, is the number of vertices contained in $e$.

Given two  hypergraphs $H_1$ and $H_2$, we say that $H_1$ is \emph{homomorphic} to $H_2$ if there is a
mapping $\varphi:V[H_1]\mapsto V[H_2]$ such that
 for any set $D\subseteq V[H_1]$, $D\in E[H_1]$ implies $\{\varphi(u): u\in D \}$ is in $E[H_2]$. We say that $H_1$ is \emph{isomorphic} to $H_2$ if the above function $\varphi$ is a bijection.
For any hypergraph $H$, the automorphism of $H$ is an isomorphism from $V[H]$ into $V[H]$. Let $\auto(H)$ be the number of automorphisms of  $H$.
For any hypergraph $H$, we call a  subgraph $H_1$ of $G$
 that is not necessarily induced
 an \emph{occurrence} of $H$, if $H_1$
is isomorphic to $H$. Let $\#(H,G)$ be the number of
occurrences of $H$ in $G$.

Let $\mathbb{S}_{\ell}$ be a permutation group of $\ell$ elements.
A $k$th root of unity is any number of the form $\ce^{2\pi\mathrm{i}\cdot j/k}$, where $0\leq j<k$.

\section{An Unbiased Estimator for Counting Hypergraphs}

Throughout the rest of the paper we assume that hypergraph $G$ has $n$ vertices and $m$ edges, and hypergraph $H$ has $t$ vertices and $k$ edges. For the notation, we denote vertices of $G$ by $u,v$ and $w$, and vertices of $H$ are denoted by $a,b$ and $c$.  For every edge $e^{\star}$ of $H$, we give the vertices in $e^{\star}$ an arbitrary ordering and call this \emph{oriented} edge $\overrightarrow{e^{\star}}$. For simplicity and with slight abuse of notation we will use $H$ to express such an oriented hypergraph.

At a high level, our estimator maintains $k$ complex variables $Z_{\overrightarrow{e^{\star}}}$, $e^{\star}\in E[H]$. These complex variables correspond to $k$ edges of hypergraph $H$, and  are set to zero initially. For every arriving edge $e\in E[G]$ with $\size(e)=\ell$, we update every $Z_{\overrightarrow{e^{\star}}}$ with $\size(e^{\star})=\size(e)$ according to
$$
Z_{\overrightarrow{e^{\star}}}(G) \leftarrow
Z_{\overrightarrow{e^{\star}}}(G) + \sum_{(\sigma(1),\ldots,\sigma(\ell))\in \mathbb{S}_{\ell}}
M_{\overrightarrow{e^{\star}}}(u_{\sigma(1)},\ldots,u_{\sigma(\ell)}),
$$
where the summation is over all possible permutations of $(1,\ldots,\ell)$, and $M_{\overrightarrow{e^{\star}}}:(V[G])^{\ell}\mapsto \mathbf{C}$ can be computed in constant time.  Hence we can rewrite $Z_{\overrightarrow{e^{\star}}}$ as
$$
Z_{\overrightarrow{e^{\star}}}(G) =\sum_{\substack{e\in E[G] \\ \size(e)=\size(e^{\star}) }} \sum_{(\sigma(1),\ldots,\sigma(\ell))\in \mathbb{S}_{\ell}}
M_{\overrightarrow{e^{\star}}}(u_{\sigma(1)},\ldots,u_{\sigma(\ell)}).
$$
Intuitively $M_{\overrightarrow{e^{\star}}}(u_{\sigma(1)},\ldots,u_{\sigma(\ell)})$ expresses the event to give edge  $e=\{u_1,\ldots,u_{\ell}\}$ in $G$ an orientation according to a permutation $(\sigma(1),\ldots, \sigma(\ell))$, and map this \emph{oriented} edge $\overrightarrow{e}$ to $\overrightarrow{e^{\star}}$. When the number of subgraph $H$ is asked, the algorithm  outputs the real part of $\alpha\cdot\prod_{\overrightarrow{e^{\star}}} Z_{\overrightarrow{e^{\star}}}$, where $\alpha\in\mathbf{R}^+$ is a scaling factor and will be determined later.

More formally, each $M_{\overrightarrow{e^{\star}}}(u_1,\ldots,u_{\ell})$ is defined according to  the degree of vertices  in graph $H$ and
 determined by three types of random variables $Q, X_{c}(w)$ and $Y(w)$, where
 $c\in V[H]$ and $w\in V[G]$: (1)  Variable $Q$ is a random $\tau$th root of unity, where $\tau:=2^t-1$. (2)  For vertex $c\in V[H], w\in V[G]$, $X_{c}(w)$ is random $\deg_H(c)$th root of unity, and for each vertex $c\in V[H]$, $X_{c}:V[G] \rightarrow \mathbf{C}$ is chosen
     independently and uniformly at random from a family of $(2t\cdot k)$-wise independent hash functions, where $2t\cdot k=O(1)$. Variables $Q$ and $X_{c}~(c\in V[H])$ are chosen independently. (3) For every $w\in V[G]$, $Y(w)$ is a random element chosen from $S:=\left\{1,2,4,8,\dots, 2^{t-1}\right\}$ as part of a $4k$-wise independent hash function. Variables $Y(w)~(w\in V[G])$ and $Q$ are chosen independently.

Given these, for every edge $\overrightarrow{e^{\star}}=(c_1,\ldots,c_{\ell})$ we define the function $M_{\overrightarrow{e^{\star}}}$ as
$$
M_{\overrightarrow{e^{\star}}}(u_1,\ldots,u_{\ell}):
=\prod_{1\leq i\leq\ell} \left(X_{c_i}(u_i)\cdot Q^{\frac{Y(u_i)}{ \deg_H(c_i) }}\right).$$
See Estimator~1 for the formal description of the update and query procedures.

\floatstyle{ruled}
\newfloat{algorithm}{htb}{loa}
\floatname{algorithm}{Estimator}
\begin{algorithm}[htb]

\vspace{0.5em}

 \underline{\textsf{Update Procedure:}} When an edge $e=\{u_1,\ldots, u_{\ell}\} \in E[G]$ arrives, update each $Z_{\overrightarrow{e^{\star}_j}}$ with $\mathrm{size}(e^{\star}_j)=\ell$ w.r.t.
\begin{align}\label{update_formula}
Z_{\overrightarrow{e^{\star}_j}}(G)\leftarrow & Z_{\overrightarrow{e^{\star}_j}}(G)+ \sum_{(\sigma(1),\ldots,\sigma(\ell))\in\mathbb{S}_{\ell}}
M_{\overrightarrow{e^{\star}_j}}\left(u_{\sigma(1)},\ldots, u_{\sigma(\ell)}\right).
\end{align}

\underline{\textsf{Query Procedure:}} When $\#(H,G)$ is required, output the real part of
\begin{equation}\label{returvalue}
\frac{t^t}{t!\cdot \auto(H)}\cdot Z_{H}(G)\enspace,
\end{equation}
where $Z_H(G)$ is defined by
\begin{equation}\label{define_Y}
Z_H(G):=\prod_{\overrightarrow{e^{\star}}\in E[H]}Z_{\overrightarrow{e^{\star}}}(G)\enspace.
\end{equation}
\caption{Counting $\#(H,G)$\label{generalalgo}}
\end{algorithm}

Before analyzing the algorithm, let us briefly discuss some properties of our algorithm.
First, the estimator runs in the turnstile model. For simplicity
we only write the update procedure for the edge insertion case. For every coming item that represents an edge-deletion, we replace ``$+$" by ``$-$" in (\ref{update_formula}).
Second, our estimator works in the distributed setting, where there are several distributed sites, and each site receives a stream $S_i$ of
hyperedges. For such settings every local site does the same for coming edges in the local stream $S_i$ .  When the number of subgraphs is asked, these  sites cooperate to give an approximation of $\#(H,G)$ for the underlying graph $G$ formed by $\bigcup_{i} S_i$.
Third, we can generalize Estimator~1 to the labelled graph case. Namely, there are labels for every vertex (and/or edge) in $G$ and $H$, and the algorithm can count the number of isomorphic copies of $H$ in $G$ whose labels are the same as $H$'s.

\section{Analysis of the Estimator}\label{sec:correctnessproof}

In this section, we first prove that  $Z_H(G)$ defined by (\ref{define_Y}) is an unbiased estimator for $\#(H,G)$. Then, we analyze the variance of the estimator and the space requirement of our algorithm in order to achieve a $(1\pm\varepsilon)$-approximation.

We first explain the intuition behind our estimator.  By
(\ref{update_formula}) and (\ref{define_Y})
 we have
\begin{align}\label{eq:expresszh}
Z_H(G)
  &= \prod_{\overrightarrow{e^{\star}}\in E[H]}\left[\sum_{
 \substack{e\in E[G]\\\mathrm{size}(e)=\mathrm{size}(\overrightarrow{e^{\star}})\\e=\{u_1,\ldots,u_{\ell}\}}}\
 \sum_{(\sigma(1),\ldots,\sigma(\ell))\in\mathbb{S}_{\ell}}
 M_{\overrightarrow{e^{\star}}}\left(u_{\sigma(1)},\ldots, u_{\sigma(\ell)} \right)\right]
 \enspace.
\end{align}
Since $H$ has $k$ edges, $Z_H(G)$  is a product of $k$ terms, and each term $Z_{\overrightarrow{e^{\star}}}(G)$ is a sum over all possible edges $e$ of $G$ with $\size(e)=\size(e^{\star})$ together with all possible orientations of $e$. Hence, in the expansion of $Z_H(G)$, any $k$-tuple $(e_1,\ldots,e_k)\in E^k(G)$ with $\size(e_i)=\size(e_i^{\star})$ contributes $\prod_{1\leq i\leq k} \left(\size(e_i)!\right)$ terms to $Z_H(G)$, and each term corresponds to a certain orientation of edges $e_1,\ldots,e_k$.

Let $\overrightarrow{T}=(\overrightarrow{e_1},\ldots, \overrightarrow{e_{k}})$ be an arbitrary orientation of $(e_1,\ldots,e_k)$, and let $G_{\overrightarrow{T}}$ be the graph induced by $\overrightarrow{T}$.
Our algorithm relies on three types of variables  to test if $G_{\overrightarrow{T}}$ is isomorphic to $H$.
These variables play different roles, as described below.
(i) For $c\in V[H]$ and $w\in V[G]$,
 we have $\ee\big[X_c^i(w)\big]\neq 0~(1\leq i\leq \deg_H(c))$ if and only if $i=\deg_H(c)$.
 Random variables $X_c(w)$ guarantee that $G_{\overrightarrow{T}}$ contributes to $\ee[Z_H(G)]$  only if $H$ is surjectively homomorphic to $G_{\overrightarrow{T}}$, i.e., $H$ is homomorphic to $G_{\overrightarrow{T}}$ and $|V_{\overrightarrow{T}}|\leq |V[H]|$.
 (ii) Through function $Y: V[G]\rightarrow S$, every vertex $u\in V_{\overrightarrow{T}}$ maps to a random element $Y(u)$ in $S$. If $|V_{\overrightarrow{T}}|=|S|=t$, then with constant probability, vertices in $V_{\overrightarrow{T}}$ map to different $t$ numbers in $S$.
Otherwise, $|V_{\overrightarrow{T}}|<t$ and vertices in $V_{\overrightarrow{T}}$ cannot map to different $t$ elements.  Since $Q$
 is a random $\tau$th root of unity, $\ee\big[Q^i\big]\neq 0\ (1\leq i\leq \tau)$ if and only if $i=\tau$, where $\tau=\sum_{\ell\in S}\ell$. The combination of $Q$ and $Y$ guarantees that $G_{\overrightarrow{T}}$ contributes to $\ee[Z_H(G)]$ only if graph $H$ and $G_{\overrightarrow{T}}$ have the same number of vertices.
Combining (i) and (ii), only subgraphs isomorphic to $H$ contribute to $\ee[Z_H(G)]$.

\subsection{Analysis of the First Moment}

Now we show that $Z_H(G)$ defined by (\ref{define_Y}) is an unbiased estimator. We first list some lemmas that we use in proving the main theorem.

\begin{lem}[\cite{conf/approx/Ganguly04}]\label{lem:primitive}
Let $X_c$ be a randomly chosen $\deg_H(c)$th root of  unity, where $c\in V[H]$. Then, for any $1< i\leq \deg_H(c)$,
it holds that $\ee\left[X_c^i\right]=1$ if $i=\deg_H(c)$, and $\ee\left[X_c^i\right]=0$ otherwise.
\end{lem}

\begin{lem}[\cite{conf/icalp/KaneMSS12}]\label{lem:rootpower}
Let $R$ be a primitive $\tau$th root of unity and  $k\in\mathbb{N}$.
If $\tau\mid k$, then $\sum_{\ell=0}^{\tau-1} (R^{k})^{\ell} =\tau$, otherwise $\sum_{\ell=0}^{\tau-1} (R^{k})^{\ell}=0$.
 \end{lem}

\begin{lem}[\cite{conf/icalp/KaneMSS12}]\label{lem:numbertheorylemma}
Let $x_i\in\mathbf{Z}_{\geq 0}$ and $\sum_{i=0}^{t-1} x_i \leq t$. Then $2^t-1 \mid \sum_{i=0}^{t-1} 2^i\cdot x_i $ if and only if $x_0=\cdots=x_{t-1}=1$.
\end{lem}

\begin{thm}\label{main_bheo}
Let $H$ be a hypergraph with $t$ vertices and $k$ edges $e^{\star}_1,\ldots,e^{\star}_k$.
 Assume that variables $X_c(w), Y(w)~(c\in V[H], w\in V[G])$ and $Q$ are defined as above.
Then,
\[\ee[Z_H(G)]=\frac{t!\cdot\auto(H)}{t^t} \cdot \#(H,G).
\]
\end{thm}

\begin{proof}
Let $q_i$ be the size of  edge $e^{\star}_i$ in $H$.
 Consider the expansion of $Z_H(G)$:
\begin{align*}
Z_H(G)
 &= \prod_{\overrightarrow{e_i^{\star}}\in E[H]}\left[\sum_{
 \substack{e\in E[G]\\\mathrm{size}(e)=\mathrm{size}({e^{\star}})\\e=\{u_1,\ldots,u_{\ell}\}}}\
 \sum_{(\sigma(1),\ldots,\sigma(\ell))\in\mathbb{S}_{\ell}}
 M_{\overrightarrow{e_i^{\star}}}\left(u_{\sigma(1)},\ldots, u_{\sigma(\ell)} \right)\right]\\
 &= \sum_{\substack{e_1,\ldots e_k\in E[G]\\
 \forall i: \mathrm{size}(e_i)=\mathrm{size}(e^{\star}_i)\\
 e_i=(u_{i,1},\ldots,u_{i,q_i})}} \sum_{\substack{ \sigma_1,\ldots, \sigma_k \\
 \forall i: \sigma_i\in\mathbb{S}_{q_i}}} \prod_{1\leq i\leq k}
 M_{\overrightarrow{e^{\star}_i}}\left(u_{i,\sigma_i(1)},
 \ldots, u_{i,\sigma_i(q_i)}\right).
\end{align*}
Hence the term corresponding to edges $e_1,\ldots, e_k$ with $\size(e_i)=\size(e^{\star}_i)$ and an arbitrary orientation $\sigma_1,\ldots,\sigma_k$ of edges $e_1,\ldots, e_k$ is
\begin{equation}\label{term_expression}
\prod_{1\leq i\leq k}
 M_{\overrightarrow{e^{\star}_i}}\left(u_{i,\sigma_i(1)},
 \ldots, u_{i,\sigma_i(\size({e^{\star}_i}) )}\right)
 = \prod_{ 1\leq i\leq k}
 \prod_{1\leq j\leq \size({e^{\star}_i})}
  X_{c^i_j}\left(w^i_j\right) Q^{\frac{Y\left(w^i_j \right)}{\deg_H\left(c^i_j \right)}}
 \enspace,
\end{equation}
where $c^i_j$ is the $j$th vertex of edge $\overrightarrow{e^{\star}_i}$, and $w^i_j$ is the $j$th vertex of edge $\overrightarrow{e_i}$.

Consider $\overrightarrow{T}=(\overrightarrow{e_1},\ldots, \overrightarrow{e_k})$ with $\size(e_i)=\size(e^{\star}_i)$, where  $\overrightarrow{e_i}$ is determined by $e_i$ and an arbitrary
orientation.
We show that the expectation of (\ref{term_expression}) is non-zero if and only if the graph induced by $\overrightarrow{T}$ is an occurrence of  $H$ in $G$.
Moreover, if the expectation of (\ref{term_expression}) is non-zero, then its value is a constant.

For a vertex $c$ of $H$ and a vertex $w$ of $G$, let
\[
 \gamma_{\overrightarrow{T}}(c,w):=\left |\big\{(i, j) ~:~ c^i_j = c \text{ and } w^i_j = w\big\} \right |
\]
be the number of pairs $(i,j)$ where the $j$th vertex of $\overrightarrow{e^{\star}_i}$ in $H$ is $c$, and the $j$th vertex of $\overrightarrow{e_i}$ in $\overrightarrow{T}$ is $w$.
Since every vertex $c$ of $H$ is incident to $\deg_H(c)$ edges, for any $c\in V[H]$, it holds that  $\sum_{w\in
V_{\overrightarrow{T}}}\gamma_{\overrightarrow{T}}(c,w)=\deg_H(c)$.
By the definition of $\gamma_{\overrightarrow{T}}$, we rewrite (\ref{term_expression}) as
\[
\left(\prod_{c\in V[H]}\prod_{w\in V_{\overrightarrow{T}}}
X_c^{\gamma_{\overrightarrow{T}}(c,w)}(w)\right)\cdot
\left(\prod_{c\in V[H]}\prod_{w\in V_{\overrightarrow{T}}}
Q^{\frac{\gamma_{\overrightarrow{T}}(c,w)Y(w)}{\deg_H(c)}}\right).
\]
 Therefore we can rewrite $Z_H(G)$ as
 \[
\sum_{\substack{e_1,\ldots e_k\in E[G]\\
 \forall i: \mathrm{size}(e_i)=\mathrm{size}(e^{\star}_i)\\
 e_i=(u_{i,1},\ldots,u_{i,q_i})}} \sum_{\substack{ \sigma_1,\ldots, \sigma_k \\
 \forall i: \sigma_i\in\mathbb{S}_{q_i} \\ \overrightarrow{T}=(\overrightarrow{e_1},\ldots,\overrightarrow{e_k})}}
 \left(\prod_{c\in V[H]}
 \prod_{w\in V_{\overrightarrow{T}}}
 X^{\gamma_{\overrightarrow{T}}(c,w)}_c(w)\right)\cdot
\left(\prod_{c\in V[H]}
 \prod_{w\in V_{\overrightarrow{T}}}
Q^{\frac{\gamma_{\overrightarrow{T}}(c,w)Y(w)}{\deg_H(c)}}\right),
\]
where the first summation is over all  $k$-tuples of  edges in $E[G]$ with $\size(e_i)=\size(e^{\star}_i)$, and
the second summation is over all possible permutations of vertices of edges $e_1,\ldots, e_k$.
By linearity of expectations of these random variables and
the assumption that $X_c(w)~(c\in V[H], w\in V[G]), Y(w)~(w\in V[G])$ and $Q$
 have sufficient independence, we have
\begin{align*}
\lefteqn{\ee[Z_H(G)]}\\
&= \sum_{\substack{e_1,\ldots e_k\in E[G]\\
 \forall i: \mathrm{size}(e_i)=\mathrm{size}(e^{\star}_i)\\
 e_i=(u_{i,1},\ldots,u_{i,q_i})}} \sum_{\substack{ \sigma_1,\ldots, \sigma_k \\
 \forall i: \sigma_i\in\mathbb{S}_{q_i} \\ \overrightarrow{T}=(\overrightarrow{e_1},\ldots,\overrightarrow{e_k})}} \left(\prod_{c \in V[H]}\ee\left[\prod_{w\in V_{\overrightarrow{T}}}X_c^{\gamma_{\overrightarrow{T}}(c,w)}(w)\right]\right)\cdot
\ee\left[\prod_{\substack{c\in V[H] \\ w\in V_{\overrightarrow{T}}}}
Q^{\frac{\gamma_{\overrightarrow{T}}(c,w)Y(w)}{\deg_H(c)}}\right].
\end{align*}
For any $\overrightarrow{T}$, let
\begin{equation}
\alpha_{\overrightarrow{T}}:=\underbrace{\left(\prod_{c \in V[H]}\ee\left[\prod_{w\in V_{\overrightarrow{T}}}X_c^{\gamma_{\overrightarrow{T}}(c,w)}(w)\right]\right)}_{A}\cdot
\underbrace{\ee\left[\prod_{c\in V[H]}\prod_{w\in V_{\overrightarrow{T}}}
Q^{\frac{\gamma_{\overrightarrow{T}}(c,w)Y(w)}{\deg_H(c)}}\right]}_B. \label{defalpha}
\end{equation}
We will next show that $\alpha_{\overrightarrow{T}}$ is either zero or a nonzero constant independent of
$\overrightarrow{T}$. The latter is the case only if
$G_T$, the undirected hypergraph induced from edge set $\overrightarrow{T}$, is isomorphic to hypergraph $H$.

First, we consider the product $A$. Assume $A\neq 0$.
Using the same technique as \cite{conf/icalp/KaneMSS12,conf/esa/ManjunathMPS11}, we construct a homomorphism from $H$ to $G_{\overrightarrow{T}}$ under the condition $A\neq 0$.
Remember that: (i) for any $c\in V[H]$ and $w\in V_{\overrightarrow{T}}$, $\gamma_{\overrightarrow{T}}(c,w)\leq\deg_H(c)$, and (ii) for any
 $c\in V[H]$, $w\in V_{\overrightarrow{T}}$ and
 $0\leq i\leq \deg_{H}(c)$,
$\ee\left[X_c^i(w)\right]\neq 0$ if and only if $i=\deg_H(c)$ or $i=0$. Therefore,
for any fixed $\overrightarrow{T}$ and $c \in V[H]$,  $\ee\left[\prod_{w\in
V_{\overrightarrow{T}}}X_c^{\gamma_{\overrightarrow{T}}(c,w)}(w)\right]\neq 0$ if and only if
$\gamma_{\overrightarrow{T}}(c,w) \in \{ 0, \deg_H(c) \}$ for all $w$.
Now, assume that $\ee\left[\prod_{w\in
V_{\overrightarrow{T}}}X_c^{\gamma_{\overrightarrow{T}}(c,w)}(w)\right]\neq 0$ for every $c\in V[H]$. Then, $\gamma_{\overrightarrow{T}}(c,w) \in \{ 0, \deg_H(c) \}$ for all $c\in V[H]$, and  $w\in V[G]$. Since $\sum_w \gamma_{\overrightarrow{T}}(c,w)
= \deg_{H}(c)$ for any $c\in V[H]$, there exists for each $c\in V[H]$  a unique vertex $w\in V_{\overrightarrow{T}}$ such that
$\gamma_{\overrightarrow{T}}(c,w)=\deg_H(c)$.
Define $\varphi_{\overrightarrow{T}}: V[H] \rightarrow V_{\overrightarrow{T}}$ as
$\varphi_{\overrightarrow{T}}(c) = w$ for the vertex $w$ satisfying $\gamma_{\overrightarrow{T}}(c,w)=\deg_H(c)$. Then, $\varphi_{\overrightarrow{T}}$ is a homomorphism, i.e.,
a set $\{u_1,\ldots,u_{\ell}\}\in E[H]$ implies $\{\varphi(u_1),\ldots,\varphi(u_{\ell})\}\in E[G_{\overrightarrow{T}}]$.
Hence, $A\neq 0$ implies  $H$ is homomorphic to $G_{\overrightarrow{T}}$, and by Lemma~\ref{lem:primitive} we have
\begin{equation}\label{eq:constant1}
\prod_{c \in V[H]}\ee\left[\prod_{w\in V_{\overrightarrow{T}}}X_c^
{\gamma_{\overrightarrow{T}}(c,w)}(w)\right]=
\prod_{c \in V[H]}\ee\left[X_c^
{\deg_H(c)}(\varphi_{\overrightarrow{T}}(c))\right]=1\enspace.
\end{equation}

Second, we consider the product $B$. We will show that, under the condition $A\neq 0$,  $G_{T}$ is an occurrence of $H$ if and only if $B\neq 0$.
Observe that
\[
\ee\left[\prod_{c\in V[H]}\prod_{w\in V_{\overrightarrow{T}}}
Q^{\frac{\gamma_{\overrightarrow{T}}(c,w)Y(w)}{\deg_H(c)}}\right]=
\ee\left[
Q^{\sum_{c\in V[H]}\sum_{w\in V_{\overrightarrow{T}}}\frac{\gamma_{\overrightarrow{T}}(c,w)Y(w)}{\deg_H(c)}}\right]\enspace.
\]

\emph{Case~1:} Assume that $G_{T}$ is an occurrence of $H$ in $G$. Then, $|V_{\overrightarrow{T}}|=|V[H]|$,  and the homomorphism $\varphi_{\overrightarrow{T}}$ constructed above is a bijection and an isomorphism. This implies that
\begin{align*}
\sum_{c\in V[H]}\sum_{w\in V_{\overrightarrow{T}}}\frac{\gamma_{\overrightarrow{T}}(c,w)\cdot Y(w)}{\deg_H(c)}
=\sum_{c\in V[H]}Y(\varphi_{\overrightarrow{T}}(c))
=\sum_{w\in V_{\overrightarrow{T}}}Y(w).
\end{align*}
Without
loss of generality, let $V_{\overrightarrow{T}}=\{w_1,\ldots,w_t\}$. By considering all possible choices of $Y(w_1),\ldots,Y(w_t)$, denoted by $y(w_1),\ldots,y(w_t)\in S$, and independence between $Q$ and $Y(w)~(w\in V[G])$, we have
 \begin{equation}\nonumber
\begin{split}
B=&\sum_{j=0}^{\tau-1}\sum_{{y(w_1),\dots,y(w_t)\in S}} \frac{1}{\tau}\
\left(\prod_{i=1}^t\Pro{Y(w_i)=y(w_i)}\right)\cdot
\mathrm{exp}\left({\frac{2\pi\mathrm{i}j}{\tau}\sum_{\ell=1}^ty(w_{\ell})}\right)\\
=&\sum_{j=0}^{\tau-1}\sum_{\substack{y(w_1),\dots,y(w_t)\in S\\ \vartheta:=y(w_1)+\dots+y(w_t), \tau \mid \vartheta}} \frac{1}{\tau}\
\left(\frac{1}{t}\right)^t\mathrm{exp}\left({\frac{2\pi\mathrm{i}}{\tau}\cdot \vartheta\cdot j}\right)\\
&\qquad+\sum_{j=0}^{\tau-1}\sum_{\substack{y(w_1),\dots,y(w_t)\in S\\  \vartheta:=y(w_1)+\dots+y(w_t), \tau \nmid \vartheta}} \frac{1}{\tau}\
\left(\frac{1}{t}\right)^t\mathrm{exp}\left({\frac{2\pi\mathrm{i}}{\tau}\cdot \vartheta\cdot j}\right)\enspace.
\end{split}
\end{equation}
Applying \lemref{rootpower} with $R=\exp\left(\frac{2\pi\mathrm{i}}{\tau}\right)$,  the second summation is zero.  Hence, by \lemref{numbertheorylemma}, we have
\begin{equation}
B=\sum_{\substack{y(w_1),\dots,y(w_t)\in S\\ \tau \mid y(w_1)+\dots+y(w_t)}}
\left(\frac{1}{t}\right)^t =
\sum_{\substack{y(w_1),\dots,y(w_t)\in S\\y(w_1)+\dots+y(w_t)=\tau}}
\left(\frac{1}{t}\right)^t =\left(\frac{1}{t}\right)^t\cdot t!= \frac{t!}{t^t}\enspace.
\label{eq:factor}
\end{equation}

\emph{Case~2:} Assume that $G_{T}$ is not an occurrence of $H$ in $G$. Then, $\varphi_{\overrightarrow{T}}$ is not a bijection, and trivially is not an isomorphism.
 Let $V_{\overrightarrow{T}}=\{w_1,\ldots,w_{t'}\}$, where $t'<t$. Then,
 there is a vertex $w\in V_{\overrightarrow{T}}$ and different $b,c\in V[H]$, such that $\varphi_{\overrightarrow{T}}(b)=\varphi_{\overrightarrow{T}}(c)=w$. As before, we have
\[
\sum_{c\in V[H]}\sum_{w\in V_{\overrightarrow{T}}}\frac{\gamma_{\overrightarrow{T}}(c,w)\cdot Y(w)}{\deg_H(c)}
=\sum_{c\in V[H]}Y(\varphi_{\overrightarrow{T}}(c))\enspace.
\]
By \lemref{numbertheorylemma},  $\tau\nmid \sum_{c\in V[H]}Y(\varphi(c))$ regardless of the choices of $Y(w_1),\ldots, Y(w_{t'})$. Hence,
 \begin{equation}\nonumber
\begin{split}
B=&\sum_{j=0}^{\tau-1}\sum_{\substack{y(w_1),\dots,y(w_{t'})\in S\\ \vartheta:=\sum_{c\in V[H]}y(\varphi_{\overrightarrow{T}}(c)),\tau \nmid \vartheta } } \frac{1}{\tau}\
\left(\frac{1}{t}\right)^{t'}\mathrm{exp}\left({\frac{2\pi\mathrm{i}}{\tau}\cdot \vartheta\cdot j}\right)=0\enspace,
\end{split}
\end{equation}
where the last equality follows from
 \lemref{rootpower} with $R=\exp\left(\frac{2\pi\mathrm{i}}{\tau}\right)$.

By \eq{constant1} and \eq{factor}, we have $\alpha_{\overrightarrow{T}} = t! \big/t^t$ if $\varphi_{\overrightarrow{T}}$ is an isomorphism, and $\alpha_{\overrightarrow{T}}=0$ otherwise.
Note that for every occurrence of $H$ in $G$, denoted by $H'$, there are $\mathrm{auto}(H)$ $\mbox{isomorphic}$ mappings between $H'$ and $H$, and each such mapping $\varphi_{\overrightarrow{T}}$ corresponds to one $T$ together with an appropriate orientation of  every edge.
Hence, every $H'$ is counted $\mathrm{auto}(H)$ times and
\[
\ee[Z_H(G)]=
\sum_{\substack{e_1,\ldots e_k\in E[G]\\
 \forall i: \mathrm{size}(e_i)=\mathrm{size}(e^{\star}_i)\\
 e_i=(u_{i,1},\ldots,u_{i,q_i})}} \sum_{\substack{ \sigma_1,\ldots, \sigma_k \\
 \forall i: \sigma_i\in\mathbb{S}_{q_i} \\ \overrightarrow{T}=(\overrightarrow{e_1},\ldots,\overrightarrow{e_k})}}
 \alpha_{\overrightarrow{T}}= \frac{t!\cdot \auto(H)}{t^t}\cdot \#(H,G)\enspace.\qedhere
\]
\end{proof}

\begin{proof}[Proof of \thmref{graphpoly}] By Theorem~\ref{main_bheo}, we have
\begin{equation}\label{graphpolynomial}
\#(H,G)=\frac{t^t}{t!\cdot\mathrm{auto}(H)}\cdot\ee[Z_H(G)].
\end{equation}
Expanding the right-hand side of (\ref{graphpolynomial}) by the definition of the expectation, the theorem holds.
\end{proof}

\subsection{Analysis of the Second Moment}

Now we  analyze the variance of $Z_H(G)$ and use Chebyshev's inequality to upper bound the space requirement of our algorithm in order to get a $(1\pm\varepsilon)$-approximation of $\#(H,G)$.
Our analysis relies on the following lemma about the number of subgraphs in a hypergraph.

\begin{lem}\label{upperboundhygraph}
Let $G$ be a hypergraph with $m$ edges, and $H$ be a hypergraph with $k$ edges and minimum degree 2. Then $\#(H,G)=O(m^{k/2})$.
\end{lem}
\begin{proof}
We define \emph{the fractional cover} $\varphi:E[H]\mapsto[0,1]$ as $\varphi(e)=1/2$ for every $e\in E[H]$. Since the minimum degree of graph $H$ is 2, we have $\sum_{e\ni v}\varphi(e) \geq 1$ for every $v\in V[
H]$. Therefore the \emph{fractional cover number}  $\min_{\varphi}\left\{\sum_{e\in E[H]} \varphi(e)\right\}\leq k/2$. By Theorem~1.1 of \cite{journals/IJM/FK98}, the lemma holds.
\end{proof}

\begin{thm}\label{thm:variance}
Let $G$ be a hypergraph with $m$ edges, and $H$ be a hypergraph with $k$ edges. Random variables $X_c(w), Y(w)~(c\in V[H], w\in V[G])$ and $Q$ are defined as above. Then the following statements hold:
(1) $\ee[Z_H(G)\cdot\overline{Z_H(G)}]=O(m^{2k})$; (2) If the minimum degree of $H$ is at least 2, then
$\ee[Z_H(G)\cdot\overline{Z_H(G)}]=O(m^k)$.
\end{thm}

\begin{proof} By definition we write $\ee[Z_H(G)\cdot\overline{Z_H(G)}]$ as
\begin{align*}
\lefteqn{\ee\left[Z_H(G)\cdot\overline{Z_H(G)}\right]} \\
=& \ee\left[\left(
\sum_{\substack{e_1,\ldots, e_k\in E[G]\\
 \forall i: \mathrm{size}(e_i)=\mathrm{size}(e^{\star}_i)\\
 e_i=(u_{i,1},\ldots,u_{i,q_i})}} \sum_{\substack{ \sigma_1,\ldots, \sigma_k \\
 \forall i: \sigma_i\in\mathbb{S}_{q_i} \\ \overrightarrow{T_1}=(\overrightarrow{e_1},\ldots,\overrightarrow{e_k})}}
\left(\prod_{\substack{c\in V[H] \\ w\in V_{\overrightarrow{T_1}}}}X^{\gamma_{\overrightarrow{T_1}}(c,w)}_c(w)\right)\cdot
\left(\prod_{\substack{c\in V[H] \\ w\in V_{\overrightarrow{T_1}}}}
Q^{\frac{\gamma_{\overrightarrow{T_1}}(c,w)Y(w)}{\deg_H(c)}}\right)\right)\cdot
\right.\\ &  \left. \overline{\left(\sum_{\substack{e'_1,\ldots, e'_k\in E[G]\\
 \forall i: \mathrm{size}(e'_i)=\mathrm{size}(e^{\star}_i)\\
 e'_i=(v_{i,1},\ldots,v_{i,q_i})}} \sum_{\substack{ \sigma'_1,\ldots, \sigma'_k \\
 \forall i: \sigma'_i\in\mathbb{S}_{q_i} \\ \overrightarrow{T_2}=(\overrightarrow{e'_1},\ldots,\overrightarrow{e'_k})}}
\left(\prod_{\substack{c\in V[H] \\ w\in V_{\overrightarrow{T_2}}}}X^{\gamma_{\overrightarrow{T_2}}(c,w)}_c(w)\right)\cdot
\left(\prod_{\substack{c\in V[H] \\ w\in V_{\overrightarrow{T_2}}}}
Q^{\frac{\gamma_{\overrightarrow{T_2}}(c,w)Y(w)}{\deg_H(c)}}\right)\right)} \right]\\
&=\ee\left[\sum_{\substack{e_1,\ldots, e_k\in E[G]\\
 \forall i: \mathrm{size}(e_i)=\mathrm{size}(e^{\star}_i)\\
 e_i=(u_{i,1},\ldots,u_{i,q_i})}} \sum_{\substack{ \sigma_1,\ldots, \sigma_k \\
 \forall i: \sigma_i\in\mathbb{S}_{q_i} \\ \overrightarrow{T_1}=(\overrightarrow{e_1},\ldots,\overrightarrow{e_k})}}
 \sum_{\substack{e'_1,\ldots, e'_k\in E[G]\\
 \forall i: \mathrm{size}(e'_i)=\mathrm{size}(e^{\star}_i)\\
 e'_i=(v_{i,1},\ldots,v_{i,q_i})}} \sum_{\substack{ \sigma'_1,\ldots, \sigma'_k \\
 \forall i: \sigma'_i\in\mathbb{S}_{q_i} \\ \overrightarrow{T_2}=(\overrightarrow{e'_1},\ldots,\overrightarrow{e'_k})}}\right.\\
 & \left.\left(\prod_{\substack{c\in V[H] \\ w\in V_{\overrightarrow{T_1}\cup\overrightarrow{T_2}}}}X^{\gamma_{\overrightarrow{T_1}}(c,w)-
 \gamma_{\overrightarrow{T_2}}(c,w)}_c(w)\right)\cdot\left(\prod_{\substack{c\in V[H] \\ w\in V_{\overrightarrow{T_1}\cup\overrightarrow{T_2}}}}
Q^{\frac{\left(\gamma_{\overrightarrow{T_1}}(c,w)-\gamma_{\overrightarrow{T_2}}(c,w)\right)\cdot Y(w)}{\deg_H(c)}}\right)\right]
\end{align*}
By linearity of expectations and the condition that random variables $X_c(w) (c\in V[H], w\in V[G])$ are $(2t\cdot k)$-wise independent, and $X_c( c\in V[H]), Q$ are chosen independently, we can rewrite $\ee[Z_H\cdot\overline{Z_H}]$ as
$$
\sum_{\substack{e_1,\ldots, e_k\in E[G]\\
 \forall i: \mathrm{size}(e_i)=\mathrm{size}(e^{\star}_i)\\
 e_i=(u_{i,1},\ldots,u_{i,q_i})}} \sum_{\substack{ \sigma_1,\ldots, \sigma_k \\
 \forall i: \sigma_i\in\mathbb{S}_{q_i} \\ \overrightarrow{T_1}=(\overrightarrow{e_1},\ldots,\overrightarrow{e_k})}}
 \sum_{\substack{e'_1,\ldots, e'_k\in E[G]\\
 \forall i: \mathrm{size}(e'_i)=\mathrm{size}(e^{\star}_i)\\
 e'_i=(v_{i,1},\ldots,v_{i,q_i})}} \sum_{\substack{ \sigma'_1,\ldots, \sigma'_k \\
 \forall i: \sigma'_i\in\mathbb{S}_{q_i} \\ \overrightarrow{T_2}=(\overrightarrow{e'_1},\ldots,\overrightarrow{e'_k})}}
 \alpha_{\overrightarrow{T_1},\overrightarrow{T_2}}
$$
where the value of $\alpha_{\overrightarrow{T_1},\overrightarrow{T_2}}$ is
\[
  {\prod_{c\in V[H]}\ee\left[\prod_{w\in V_{\overrightarrow{T_1}\cup\overrightarrow{T_2}}}X^{\gamma_{\overrightarrow{T_1}}(c,w)-
 \gamma_{\overrightarrow{T_2}}(c,w)}_c(w)\right]\cdot\ee\left(\prod_{\substack{c\in V[H] \\ w\in V_{\overrightarrow{T_1}\cup\overrightarrow{T_2}}}}
Q^{\frac{\left(\gamma_{\overrightarrow{T_1}}(c,w)-\gamma_{\overrightarrow{T_2}}(c,w)\right)\cdot Y(w)}{\deg_H(c)}}\right)}=O(1).
\]
Since $\ee[Z_H(G)\cdot\overline{Z_H(G)}]$ has at most $O(m^{2k})$ terms, the first statement holds.

Now for the second statement.
Remember that (i) for any $c\in V[H]$ and $w\in V_{\overrightarrow{T_1}\cup\overrightarrow{T_2}}$, $\ee[X_c^i(w)]\neq 0$ if and only if $i$ is divisible by $\deg_H(c)$, and (ii) for any $c\in V[H]$ and $w\in V_{\overrightarrow{T_1}\cup\overrightarrow{T_2}
}$, it holds that $0\leq \gamma_{\overrightarrow{T_1}}(c,w)\leq \deg_H(c)$ and $0\leq \gamma_{\overrightarrow{T_2}}(c,w)\leq \deg_H(c)$.
Hence $\alpha_{\overrightarrow{T_1},\overrightarrow{T_2}}\neq 0$ if  for any $c\in V[H]$ and $w\in V[G]$ it holds that (i) $\gamma_{\overrightarrow{T_1}}(c,w)=\gamma_{\overrightarrow{T_2}}(c,w)$,
or  (ii) $\gamma_{\overrightarrow{T_1}}(c,w)=\deg_H(c)$,
$\gamma_{\overrightarrow{T_2}}(c,w)=0 $, or (iii) $\gamma_{\overrightarrow{T_1}}(c,w)=0$,
$\gamma_{\overrightarrow{T_2}}(c,w)=\deg_H(c)$.
We partition $V_{\overrightarrow{T_1}\cup \overrightarrow{T_2}}$ into three disjoint subsets $A$, $B$ and $C$ defined by
$
A:= V_{\overrightarrow{T_1}}\setminus V_{\overrightarrow{T_2}}$,
$B:= V_{\overrightarrow{T_2}}\setminus V_{\overrightarrow{T_1}}$, and
$C:= V_{\overrightarrow{T_1}}\cap V_{\overrightarrow{T_2}}$.
Set $A$, $B$, and $C$ are defined according to the above conditions (i), (ii) and (iii). By the assumption that the minimum degree of $H$ is 2, the degree of every vertex in sets $A, B$ and $C$ is at least 2. Since there are $O(1)$ different such $H'$ of constant size, and for each $H'$ of them it holds that $\#(H,G)=O(m^{k/2})$, by Lemma~\ref{upperboundhygraph} we have $\ee[Z_H(G)\cdot \overline{Z_H(G)}]=O(m^k)$.
\end{proof}

By applying Chebyshev's inequality, we can get a $(1\pm \varepsilon)$-approximation by running our estimator in parallel and returning the average of the output of these returned values, and this implies our main theorem~(Theorem~\ref{thm:main}).

\begin{proof}[Proof of Theorem~\ref{thm:main}]
We run $s$ parallel and independent copies of our estimator and take the average value $Z^*=\frac{1}{s}\sum_{i=1}^sZ_i$,
where each $Z_i$ is the output of the $i$th instance of the estimator. Therefore, $\ee[Z^*]=\ee[Z_H(G)]$, and a straightforward calculation shows that\[
\ee\left[Z^*\overline{Z}^{*}\right]-\left|\ee\left[Z^{*}\right]\right|^{2}=
\frac{1}{s}\left(\ee\left[Z_H(G)\cdot\overline{Z_H(G)}\right]-|\ee[Z_H(G)]|^{2}\right)\enspace.\]
By Chebyshev's inequality for complex-valued random variables (see, e.g., \cite[Lemma~3]{conf/esa/ManjunathMPS11}), we have
$$\Pro{\left|Z^*-\ee[Z^*]\right| \geq \varepsilon\cdot|\ee[Z^*]|} \leq \frac{\ee\left[Z_H(G)\cdot\overline{Z_H(G)}\right]-\ee[Z_H(G)]
\cdot\overline{\ee[Z_H(G)]}}{s\cdot\varepsilon^{2}\cdot|\ee[Z_H(G)]|^{2}}\enspace.$$
By the first statement of \thmref{variance},  we have $$\ee\left[Z_H(G)\cdot\overline{Z_H(G)}\right]-\ee[Z_H(G)]\cdot
\overline{\ee[Z_H(G)]} \leq \ee\left[Z_H(G) \cdot\overline{Z_H(G)} \right]=O(m^k)\enspace.$$ By choosing $s=O\left(\frac{1}{\varepsilon^2}\cdot\frac{m^k}{(\# H)^2}\right)$, we get \[\Pro{\left|Z^*-\ee\left[Z^*\right]\right| \geq \varepsilon\cdot|\ee[Z^*]|} \leq 1/3\enspace . \]
Hence, the overall space complexity is
$O\left(\frac{1}{\varepsilon^2}\cdot\frac{m^k}{(\#H)^2}\cdot\log n\right)$.
\end{proof}

\vspace{1em}

\textbf{Acknowledgement.}  The author would like to thank Kurt Mehlhorn for helpful comments on the presentation.

\bibliographystyle{plainnat}
\bibliography{reference}

\end{document}